\documentclass{amsart}
\usepackage{amsmath}
\usepackage{amssymb}
\usepackage{mathrsfs}
\usepackage{amsfonts}
\usepackage{graphicx}
\usepackage{tikz}
\usepackage{texdraw}
\usepackage{graphpap}
\usepackage{enumitem}
\usepackage{chngcntr}
\usepackage{wasysym}
\usepackage{color}
\usepackage[OT2,T1]{fontenc}
\usepackage{tikz}
\usepackage{pgfplots}

\usetikzlibrary{decorations.markings}
\pgfplotsset{width=7cm,compat=1.8}
\usetikzlibrary{spy}

\newcommand{\boundellipse}[3]
{(#1) ellipse (#2 and #3)
}

\theoremstyle{plain}
\newtheorem{thm}{Theorem}[section]
\newtheorem{mth}{Theorem}

\newtheorem{cor}[thm]{Corollary}
\newtheorem{lem}[thm]{Lemma}
\newtheorem*{lemm}{Lemma}

\theoremstyle{remark}

\newtheorem*{rem}{Remark}

\numberwithin{equation}{section}

\newcommand{\Ext}{\operatorname{Ext}}
\newcommand{\fif}{\fii}

\newcommand{\gz}{\ell_0}
\newcommand{\gone}{\ell_1}
\newcommand{\gtau}{\ell_\tau}

\newcommand{\bfR}{\mathbf{R}}

\newcommand{\T}{{\mathbb T}}

\newcommand{\calH}{{\mathscr H}}

\newcommand{\calK}{{\mathscr K}}

\newcommand{\calQ}{{\mathscr Q}}

\newcommand{\calW}{{\mathscr W}}

\newcommand{\bfK}{{\mathbf K}}

\newcommand{\bfk}{{\mathbf k}}
\newcommand{\Expe}{{\mathbb E}}

\newcommand{\Int}{\operatorname{Int}}

\newcommand{\eqpot}{\check{Q}}

\newcommand{\erfc}{\operatorname{erfc}}

\newcommand{\R}{{\mathbb R}}
\renewcommand{\T}{{\mathbb T}}
\renewcommand{\L}{{\mathbb L}}

\newcommand{\D}{{\mathbb D}}

\newcommand{\C}{{\mathbb C}}

\newcommand{\fii}{{\varphi}}
\newcommand{\const}{\mathrm{const.}}

\newcommand{\re}{\operatorname{Re}}
\newcommand{\im}{\operatorname{Im}}

\newcommand{\Pol}{\operatorname{Pol}}
\newcommand{\degree}{\operatorname{degree}}
\newcommand{\Pc}{\operatorname{Pc}}

\newcommand{\Prob}{{\mathbb{P}}}

\renewcommand{\d}{{\partial}}
\newcommand{\dbar}{\bar{\partial}}
\newcommand{\1}{\mathbf{1}}

\newcommand{\dist}{\operatorname{dist}}

\newcommand{\Lap}{\Delta}

\begin{document}

\title[Normal matrix ensembles at the hard edge]{A note on normal matrix ensembles at the hard edge}

\begin{abstract} 
We investigate how the theory of quasipolynomials due to Hedenmalm and Wennman works in a hard edge setting and
obtain as a consequence a scaling limit for radially symmetric potentials.
\end{abstract}


\author{Yacin Ameur}

\address{Yacin Ameur\\
Department of Mathematics\\
Faculty of Science\\
Lund University\\
P.O. BOX 118\\
221 00 Lund\\
Sweden}

\email{Yacin.Ameur@maths.lth.se}

\maketitle

\section{Introduction and main result}

In the theory of Coulomb gas ensembles, it is natural to consider two different kinds of boundary conditions: the "free boundary", where particles are admitted to be outside of the droplet,
and the "hard edge", where they are completely confined to it. In this note we consider the determinantal, two-dimensional hard edge case, i.e., we consider
eigenvalues of random normal matrices with a hard edge spectrum.

Hard edge ensembles are well-known in the Hermitian theory, where they are usually associated with the Bessel kernel \cite{Fo,Fo2,TW2}. Another possibility, a "soft/hard edge", appears when a soft edge is replaced by a hard edge cut. This situation was studied by Claeys and Kuijlaars in the paper \cite{CK}. The case at hand is also of the soft/hard type, but
to keep our terminology simple, we prefer to use the adjective "hard".

For the hard edge Ginibre ensemble,
a direct computation with the orthogonal polynomials in \cite[Section 2.3]{AKM} shows that,
under a natural scaling about a boundary point, the point process
of eigenvalues converges to the determinantal point field determined by the 1-point function
\begin{equation}\label{unih}R(z)=H(z+\bar{z})\cdot \1_\L(z),\end{equation}
where $\1_\L$ is the indicator function of the left half plane $\L=\{\re z<0\}$ and where $H$, the "hard edge plasma function", is defined by
\begin{equation}\label{hfun} H(z)=\frac 1 {\sqrt{2\pi}}\int_{-\infty}^0\frac {e^{-(z-t)^2/2}}{\fif(t)}\, dt.\end{equation}
Here $\fif$, the "free boundary function", is given by
\begin{equation}\label{ffun}\fif(t)=\frac 1 {\sqrt{2\pi}}\int_t^{+\infty}e^{-s^2/2}\, ds=\frac 1 2 \erfc\frac t{\sqrt{2}}.\end{equation}
As far as we know, the density profile \eqref{unih} appeared first in the physical paper \cite{S} from 1982, cf. \cite[Section 15.3.1]{Fo};
see Figure 1.

 In the paper \cite{HW}, it is shown that if free boundary conditions are imposed, then the point field with intensity function $R(z)=\fif(z+\bar{z})$ appears universally (i.e. for a "sufficiently large" class of ensembles) when rescaling about a regular boundary point. In the present article, we shall discuss (but not completely solve) the problem of showing
 that the 1-point function \eqref{unih} appears universally at a hard edge.

\begin{figure}[ht]
\begin{center}
\includegraphics[width=.25\textwidth]{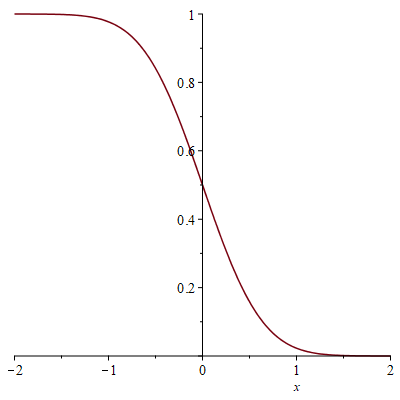}
\hspace{.15\textwidth}
\includegraphics[width=.25\textwidth]{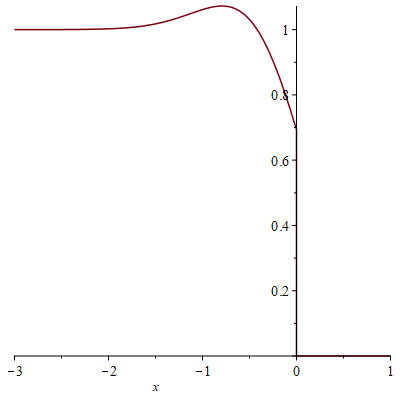}
\end{center}
\caption{Boundary profiles $R(x)$ at the free boundary and the hard edge, respectively.}
\end{figure}

Our method uses an adaptation the analysis of the orthogonal polynomials for the free boundary case in the recent paper \cite{HW}. More precisely, we shall modify the quasipolynomials from \cite{HW} so that they suit for hard edge conditions.
When this is carried out, the $H$-function appears after rescaling about a given boundary point, by taking a simple-minded sum of weighted quasipolynomials, and passing to the large volume limit.
The latter limit is quite general, but due to technical difficulties, we shall ultimately require that the external potential be radially symmetric,
in order to obtain a complete description of the limiting point field. However, we have chosen to present our approach in a general setting, hoping
that the method can contribute to a future clarification of the general case.



\subsection{Basic setup}
Fix a function ("external potential") $Q:\C\to\R\cup\{+\infty\}$ and write $\Sigma=\{Q<+\infty\}$. We assume that $\Int \Sigma$ be dense in $\Sigma$ and that $Q$ be lower semicontinuous on $\C$ and real-analytic on $\Int\Sigma$ and "large" near $\infty$:
$\liminf_{\zeta\to\infty}Q(\zeta)/\log|\zeta|^2>1.$

We next form the \textit{equilibrium measure} $\sigma$ in external potential $Q$, namely the measure $\mu$ that minimizes the weighted energy
\begin{equation}\label{q-en}I_Q[\mu]=\iint_{\C^2}\log\frac 1 {|\zeta-\eta|}\, d\mu(\zeta)d\mu(\eta)+\int_\C Q\, d\mu\end{equation}
amongst all compactly supported Borel probability measures $\mu$ on $\C$.
It is well-known \cite{ST} that $\sigma$ exists, is absolutely continuous, and takes the form
\begin{equation}\label{EQMEAS}d\sigma(\zeta)=\Lap Q(\zeta)\cdot\1_S(\zeta)\, dA(\zeta),\end{equation}
where $S$ is a compact set which we call the \textit{droplet} in potential $Q$. Here and henceforth we use the convention that $\Lap=\d\dbar$ denotes $1/4$ times the usual Laplacian, while $dA=dxdy/\pi$ is Lebesgue measure normalized so that the unit disk has measure $1$.

We will in the following assume that $S\subset\Int\Sigma$ and that $S$ be connected.
Under our assumptions, $S$ is finitely connected and the boundary $\d S$ is a union of a finite number of real-analytic arcs, possibly with finitely many singular points which can be certain types of cusps and/or double points. See e.g. \cite{AKMW,LM}.

We shall consider the \textit{outer boundary}
$\Gamma=\d \Pc S,$
where the \textit{polynomially convex hull} $\Pc S$ is the
union of $S$ and the bounded components of $\C\setminus S$.
Thus $\Gamma$ is a Jordan curve having possibly finitely many singular points. We shall assume that $\Gamma$ be \textit{everywhere regular}, i.e., that there are no singular points on $\Gamma$.

\subsection{Hard edge ensembles} Given a potential $Q$ satisfying our above assumptions, we define the corresponding \textit{localized} potential
$Q^S=Q+\infty\cdot \1_{\C\setminus S},$
i.e., $Q^S$ equals to $Q$ in $S$ and to $+\infty$ in the complement $\C\setminus S$.
Consider the Gibbs probability measure on $\C^n$ in external field $nQ^S$
\begin{equation}\label{gibb}d\Prob_n^{}\propto e^{-nH_n},\,
\qquad H_n(\zeta_1,\ldots,\zeta_n)=\sum_{j\ne k}\log\frac 1 {|\zeta_j-\zeta_k|}+n\sum_{j=1}^nQ^S(\zeta_j).\end{equation}
Here $dV_n=dA^{\otimes n}$ is Lebesgue measure on $\C^n$ divided by $\pi^n$;
"$\propto$" indicates a proportionality.

The corresponding eigenvalue ensemble $\{\zeta_j\}_1^n$ is just a configuration picked randomly with respect to \eqref{gibb}.
As in the free boundary case, the system $\{\zeta_j\}_1^n$ roughly tends to follow the equilibrium distribution, in the sense that
$\frac 1 n \Expe_n[f(\zeta_1)+\cdots +f(\zeta_n)]\to \sigma (f)$ as $n\to\infty$ for each bounded continuous function $f$.

The point process $\{\zeta_j\}_1^n$ is determined by the collection of $k$-point functions $\bfR_{n,k}$.
In the case at hand, we have the basic determinant formula
$\bfR_{n,k}^{}(\eta_1,\ldots,\eta_k)=\det(\bfK_n^{}(\eta_i,\eta_j))_{i,j=1}^k,$
where the \textit{correlation kernel} $\bfK_n^{}$ can be taken as the reproducing kernel for the subspace $\calW_n^{}$ of $L^2=L^2(\C,dA)$ consisting of all "weighted polynomials" $w=pe^{-nQ^S/2}$ where $p$ is an analytic polynomial of degree at most $n-1$. (This \textit{canonical} correlation kernel is used
without exception below.)

We write $\bfR_n^{}=\bfR_{n,1}^{}$ for the $1$-point function, which is the key player in our discussion.

\subsection{Scaling limit}
Let us now fix a (regular) point on the outer boundary $\Gamma$, without loss we place it at the origin, and so that the outwards normal to the boundary points in the positive real direction.

We define a \textit{rescaled process} $\{z_j\}_1^n$ by magnifying distances about $0$ by a factor $\sqrt{n\Lap Q(0)}$,
$$z_j=\sqrt{n\Lap Q(0)}\,\zeta_j,\qquad\qquad (j=1,\ldots,n).$$

In general, we will denote by $\zeta,z$ two complex variables related by $z=\sqrt{n\Lap Q(0)}\, \zeta$.
We regard the droplet $S$ as a subset of the $\zeta$-plane.
 Restricting to a fixed bounded subset of the $z$-plane, the image of the droplet then more and more resembles left half plane $\L$, as $n\to\infty$; see Figure 2.

\begin{figure}[ht]
\begin{tikzpicture}[spy using outlines=
	{circle, magnification=10, connect spies}]
\draw
\boundellipse{0,0}{1}{1.61803};
\draw[gray,fill] \boundellipse{0,0}{1}{1.61803};
\node at (1,1.5) {$(\zeta)$};
\node at (6.5,1.25) {$(z)$};
  \coordinate (spypoint) at (1,0);
  \coordinate (magnifyglass) at (5,0);
  \spy [black, size=3cm] on (spypoint)
   in node[fill=white] at (magnifyglass);
\end{tikzpicture}
\caption{Rescaling about a boundary point.}
\end{figure}
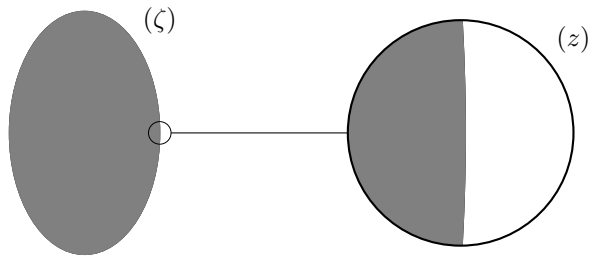

Following \cite{AKM} we denote by plain symbols $R_n^{},K_n^{}$, etc., the $1$-point function, canonical correlation kernel, etc., with respect to the process $\{z_j\}_1^n$. Here the canonical kernel $K_n$ is, by definition
\begin{equation*}K_n^{}(z,w)=\frac 1 {n\Lap Q(0)}\bfK_n^{}(\zeta,\eta),\quad z=\sqrt{n\Lap Q(0)}\,\zeta,\quad w=\sqrt{n\Lap Q(0)}\,\eta.\end{equation*}

Recall that a function of the form $c(\zeta,\eta)=g(\zeta)\bar{g}(\eta)$, where $g$ is a continuous unimodular function, is called a \textit{cocycle}. A function $h(z,w)$ is \textit{Hermitian} if $\bar{h}(w,z)=h(z,w)$, and \textit{Hermitian-analytic} if furthermore $h$ is analytic in $z$ and in $\bar{w}$. Finally, the \textit{Ginibre kernel} is
$G(z,w)=e^{-|z|^2/2-|w|^2/2+z\bar{w}}.$
This is the correlation kernel of the infinite Ginibre ensemble, which emerges by rescaling about a regular bulk point, see e.g. \cite{AKM}.

The following basic lemma gives the existence of subsequential limiting point fields.

\begin{lemm} \label{mlem} (i) There exists a sequence of cocycles $c_n$ such that each subsequence of the sequence $(c_nK_n)$ has a subsequence converging locally uniformly in $\L\times\L$
to a Hermitian limit $K$.

(ii) Each limiting kernel $K$ in (i) is of the form
$K=G\cdot\Psi\cdot \1_\L\otimes\1_\L$ where $G$ is the Ginibre kernel and $\Psi$ is a Hermitian-analytic function on $\L\times\L$.
\end{lemm}

\begin{proof}[Remark on the proof]
One can simply repeat what is done in the full plane case in \cite[Section 3]{AKM}, by using the potential $Q^S$ in place of $Q$ and replacing "$\C$" by "$\L$".
\end{proof}

 By polarization and the general theory of point fields (see \cite{AKMW}), the lemma implies that a limiting 1-point function $R(z)=K(z,z)$ determines a unique limiting determinantal point field $\{z_j\}_1^\infty$  with $k$-point intensity given by the determinant
$\det(K(z_i,z_j))_{i,j=1}^k.$

We are now ready to state our main result. (Recall that $H$ is the function \eqref{hfun}.)

\begin{mth}\label{mthm2} If $Q$ is radially symmetric,  then
$R_n\to R$ with bounded almost everywhere convergence on $\C$ and locally uniform convergence on $\C\setminus(i\R)$ where
$R(z)=H(z+\bar{z})\cdot \1_\L(z)$.
\end{mth}

\subsection{Comments} It is natural to conjecture that Theorem \ref{mthm2} should remain true also without supposing any special symmetries of the potential $Q$.
To support this, we briefly compare with the analysis of translation invariant limiting $1$-point functions $R$ from the
paper \cite{AKMW}. It
is there noted that any limiting 1-point function $R$ gives rise to a solution to Ward's equation in $\L$,
\begin{equation}\label{ward}\dbar C=R-1-\Lap \log R.\end{equation}
It is natural to hypothesize that the true intensity function $R$ should be vertically translation invariant, namely $R(z+it)=R(z)$ for all $z\in\L$ and all $t\in\R$. This is equivalent
to assuming that the holomorphic kernel $\Psi$ in Lemma \ref{mlem} be of the form $\Psi(z,w)=\Phi(z+\bar{w})$, i.e., $R(z)=\Phi(z+\bar{z})$.
It is natural, in addition, to assume that $\Phi$ is of
"error-function type", in the sense that
\begin{equation}\label{erf}\Phi(z)=\gamma*g(z):=\frac 1 {\sqrt{2\pi}}\int_{-\infty}^\infty e^{-(z-t)^2/2}g(t)\, dt\end{equation}
where $g$ is some function defined on $\R$. (The notation "$\gamma*g$" is explained in Section \ref{not} below.)

Assuming the structure $R(z)=\Phi(z+\bar{z})$ with a $\Phi$ of the form \eqref{erf}, the analysis in \cite{AKMW} shows that the only feasible solutions to \eqref{ward} are of the form
$R(z)=H(z+\bar{z}+c)\cdot \1_\L(z),$
where $c$ is some real constant. This argument makes it plausible that radial symmetry should not be needed, notwithstanding that it does not reveal the right value of $c$.
However, the problem of proving apriori translation invariance remains open at the time of writing.


Even disregarding the symmetry assumption on $Q$, our arguments break down if there is a singular point on the outer boundary.

For certain types of potentials, having finitely many logarithmic singularities, we can obtain universality also at inner boundary components.
A proof can for example be accomplished by applying the method of "root functions" from \cite{W}.

 The paper \cite{Se} gives an application of hard edge theory to the distribution of the largest modulus of an eigenvalue for radially symmetric hard edge ensembles.

\subsection{Notation} \label{not} As in \cite{AKM}, we write
$$\gamma(z)=\frac 1 {\sqrt{2\pi}}e^{-z^2/2}$$
for the "Gaussian kernel". For functions $g\in L^\infty(\R)$ we shall use the convolution by $\gamma$ in the
sense given in \eqref{erf}.
In this notation, the functions $\fif$ and $H$ become
$$\fif:=\gamma*\1_{(-\infty,0)},\qquad H:=\gamma *\left(\frac {\1_{(-\infty,0)}}\fif\right).$$

We write $\hat{\C}=\C\cup\{\infty\}$. By $D(a,\rho)$ we denote the open disc in $\C$ with center $a$, radius $\rho$;
 we abbreviate $\D(\rho)=D(0,\rho)$ and $\D_e(\rho)=\hat{\C}\setminus \overline{\D}_\rho$. When $\rho=1$ we simply
write $\D$ and $\D_e$. We write $\T=\d \D$, and $\Pol(j)$ denotes the linear space of analytic polynomials of degree at most $j$. The symbol $\|\cdot\|_p$ will denote the norm in $L^p:=L^p(\C,dA)$. More generally, if $\omega\ge 0$ is a measurable "weight-function", we denote by $L^p(\omega)$ the space defined
by the semi-norm $\|f\|_{L^p(\omega)}=\|f\omega^{1/p}\|_p$. We will frequently use the number $\delta_n:=n^{-1/2}\log n$, ($n\ge 2$).

We denote by "$\sim$" various asymptotic relations, usually as $n\to\infty$. The exact meaning will be clear from the context. We will denote by the same symbols "$C,c$" various unspecified
numerical constants (with $c>0$) which can change meaning from time to time, even within the same calculation.

\smallskip

For convenience of the reader, we now list some (nonstandard) notation used throughout.

\smallskip

\begin{tabular}{lll}
"Arclength":& $ds(\zeta)=\frac 1 {2\pi}|d\zeta|$ & (so  $\int_\T ds=1$)\cr
"Area":& $dA(\zeta)=\frac {d^2\zeta}\pi$ & (so \quad $\int_\D dA=1$)\cr
"Laplacian": &$\Delta=\d\dbar$ & (so  $\Lap_\zeta|\zeta|^2=1$)\cr
\end{tabular}

\section{Some preparations}
In this section, we briefly outline of our strategy for proving Theorem \ref{mthm2}, providing also some necessary background on Laplacian growth and weighted polynomials.

\subsection{Outline of strategy}

 Let $p_{j,n}$ be the $j$:th orthonormal polynomial with respect to the weight $e^{-nQ^S}$, and write $w_{j,n}=p_{j,n}e^{-nQ^S/2}$. We start with the basic identity
$$\bfR_n(\zeta)=\bfK_n(\zeta,\zeta)=\sum_{j=0}^{n-1}|w_{j,n}(\zeta)|^2.$$
As a preliminary step, we shall prove that if $\zeta$ is very close to the outer boundary
$\Gamma=\d \Pc S$, then all terms but the last $\sqrt{n}\log n$ ones can be neglected. More precisely, if $\zeta$ belongs to a \textit{belt}
\begin{equation}\label{belt}N_\Gamma=\{\zeta\in S;\, \dist(\zeta,\d S)\le C/\sqrt{n}\},\end{equation}
then
\begin{equation}\label{redu}\bfR_n(\zeta)\sim \sum_{j=n-\sqrt{n}\log n}^{n-1}|w_{j,n}(\zeta)|^2.\end{equation}

The modest goal of this section is to prove the reduction \eqref{redu}.
Later on, we will analyze the sum in \eqref{redu} in a microscopic neighbourhood of the given point $0\in \Gamma$.

\subsection{Preliminaries on Laplacian growth}

In this subsection we recast some results on Laplacian growth; this is convenient, if not otherwise, to introduce some necessary notation. References and further reading can be found, for instance in \cite{GM,HM,LM,VE} and in \cite[Section 2]{HW}.

For a parameter $\tau$ with $0<\tau\le 1$, we let
$\check{Q}_\tau$ be the \textit{obstacle function} subject to the growth condition
$$\check{Q}_\tau(\zeta)=2\tau\log|\zeta|+O(1),\quad \text{as}\quad \zeta\to\infty.$$
The precise definition runs as follows: for each $\eta\in\C$, $\check{Q}_\tau(\eta)$ is the supremum of
$s(\eta)$ where $s$ is a subharmonic function which is everywhere $\le  Q$ and satisfies $s(\zeta)\le 2\tau\log|\zeta|+O(1)$ as $\zeta\to\infty$.

Write $S_\tau$ for the droplet in external potential $Q/\tau$ and note that
$\sigma(S_\tau)=\tau$
where $\sigma$ is the equilibrium measure \eqref{EQMEAS}.
Under our conditions,
$S_\tau$ equals to closure of the interior of the coincidence set $\{Q=\check{Q}_\tau\}$ and the measure $\sigma_\tau:=\1_{S_\tau}\cdot \sigma$ minimizes the weighted energy \eqref{q-en} amongst positive measures of total mass $\tau$.

Clearly the droplets $S_\tau$ increase with $\tau$. The evolution of the $S_\tau$'s is known as Laplacian growth.
We will write $\Gamma_\tau$ for the outer boundary,
$$\Gamma_\tau=\d\Pc S_\tau.$$
Hence $\Gamma_\tau=\d U_\tau$ where
$$U_\tau=\hat{\C}\setminus \Pc S_\tau.$$
Finally, we denote by
$$\phi_\tau:U_\tau\to\D_e$$
the unique conformal (surjective) map, normalized so that $\phi_\tau(\infty)=\infty$ and $\phi_\tau'(\infty)>0$.

It is well-known that $\check{Q}_\tau$ is $C^{1,1}$-smooth on $\C$ and harmonic on $U_\tau$. Moreover, since $\Gamma=\Gamma_1$ is everywhere regular, it follows from standard facts about Laplacian growth that $\Gamma_\tau$ is everywhere regular for all $\tau$ in some interval $\tau_0\le \tau\le 1$, where $\tau_0<1$. Below we fix, once and for all, such a $\tau_0$.

The following is a variant of "Richardson's lemma" \cite{HM,VE}.

\begin{lem} If $0<\tau<\tau'\le 1$ and if $h$ is harmonic in $U_\tau$ and smooth up to the boundary, then
$$\int_{S_{\tau'}\setminus S_\tau}h\Lap Q\, dA=(\tau'-\tau)h(\infty).$$
\end{lem}

\begin{proof} Since $\sigma(S_\tau)=\tau$,
the asserted identity is true when $h$ is a constant. Hence we can assume that $h(\infty)=0$. Let us fix such a $h$ and extend it to $\C$ in a smooth way.
It follows from the properties of the obstacle function $\check{Q}_\tau$ that
$$\int_{S_\tau}h\Lap Q\, dA=\int_\C h\Lap \check{Q}_\tau\, dA=\int_\C \Lap h\cdot \1_{S_\tau}Q\, dA.$$
Subtracting the corresponding identity with $\tau$ replaced by $\tau'$ we find
$\int_{S_{\tau'}\setminus S_\tau}h\Lap Q\, dA=0,$
finishing the proof of the lemma.
 \end{proof}

The lemma says that if $d\sigma_\tau=\Lap Q\cdot\1_{S_\tau}\, dA$ is equilibrium measure of mass $\tau$, then near the outer boundary component $\Gamma_\tau$, we have, in a suitable "weak" sense that
$\frac {d\sigma_\tau} {d\tau}=\omega_\tau$
where $\omega_\tau$ is the harmonic measure of $U_\tau$ evaluated at $\infty$. This means: if $h$ is a continuous function on $\d S_\tau$, then $\omega_\tau(h)=\tilde{h}(\infty)$, where $\tilde{h}$ is the harmonic extension to $U_\tau$ of $h$.

Let us define the Green's function of $U_\tau$ with pole at infinity by
$G(\zeta,\infty)=\log|\phi_\tau(\zeta)|^2$.
Then by Green's identity,
$$\tilde{h}(\infty)=\int_{U_\tau}\tilde{h}(\zeta)\,\Lap G (\zeta,\infty)\, dA(\zeta)=-\frac 1 2 \int_{\d S_\tau} h\frac {\d G}{\d n}\, ds_\tau,$$
where $ds_\tau$ stands for the arclength measure on $\Gamma_\tau$ divided by $2\pi$.

Now for values of $\tau$ such that $\Gamma_\tau$ is everywhere regular
$-\frac \d {\d n} G(\zeta,\infty)=2|\phi'_\tau(\zeta)|$ when $\zeta\in \Gamma_\tau$
so we conclude that
$\frac d {d\tau}(\1_{S_\tau}\Lap Q\, dA)=|\phi'_\tau|ds_\tau$,
meaning that the outer boundary $\Gamma_\tau$ moves in the direction normal to $\Gamma_\tau$, at local speed $|\phi'_\tau|/2\Lap Q$. (The factor $2$ comes about because of
the different normalizations of $dA$ and $ds_\tau$.)

The above dynamic of $\Gamma_\tau$ is of course deduced in a "weak" sense, but using the regularity of the curves involved, one can turn it into a pointwise
estimate. More precisely, one has the following result, which is essentially \cite[Lemma 2.3.1]{HW}.

\begin{lem} \label{movin} Let $0$ be a boundary point of $\Gamma=\Gamma_1$ such that the outwards normal points in the positive real direction
and fix $\tau$ with $\tau_0\le\tau\le 1$. Let $x_\tau\in\Gamma_\tau\cap\R$ be the point closest to $0$. Then
$$x_{\tau}=-\frac {|\phi_1'(0)|}{2\Lap Q(0)}(1-\tau)+O((1-\tau)^2).$$
A similar estimate is true at each $\zeta\in\Gamma$ with a $O$-constant uniform in $\zeta$.
In particular there are constants $c_1,c_2>0$ such that for all $\zeta\in \Gamma$ and all $\tau_0\le\tau\le 1$,
$$c_1(1-\tau)\le \dist(\Gamma_\tau,\zeta)\le c_2(1-\tau).$$
\end{lem}

\subsection{Weighted polynomials}

Aiming to adapt the rescaling procedure in \cite{AKM} to hard edge conditions, the first thing we note is that the standard apriori estimates in \cite[Section 3]{AKM} are not directly applicable.
We here give a few supplementary estimates which will do for our present applications.

\begin{lem} \label{propp} Suppose that $Q^S$ is finite and $C^2$-smooth in a disk $D=D(\zeta,c/\sqrt{n})$. Let $u$ be a holomorphic function on $D$, and put $w=ue^{-nQ^S/2}$.
Then there is a constant $C$ depending only on $\sup_D\Lap Q$ such that $$|w(\zeta)|^2\le \frac {Cn} {c^2}\int_D|w|^2\, dA.$$
\end{lem}

\begin{proof} This follows from a standard argument, using that $\zeta\mapsto |w(\zeta)|^2e^{Kn|\zeta|^2}$ is
logarithmically subharmonic in $D$ for large enough $K$. See for instance \cite[Section 3]{AKM}.
\end{proof}

In the sequel we recall that the symbol "$\tau_0$" denotes a fixed number $<1$ such that the curves $\Gamma_\tau$ are
analytic Jordan curves when $\tau_0\le\tau\le 1$.

\begin{lem} \label{l6} Consider a weighted polynomial $w=qe^{-nQ^S/2}$.

\begin{enumerate}[label=(\roman*)]
\item \label{1st} Let $a>0$ be a constant and suppose that $q\in \Pol(j)$ where $j$ is in the range
$$n(1-\tau_0)\le j\le n-a\sqrt{n}.$$
Then $$|w(\zeta)|\le C\sqrt{n}\|w\|_{2}e^{-n(Q^S-\eqpot_\tau)(\zeta)/2}$$ where $C$ depends only on $a$.
\item \label{2nd} 
If $q\in\Pol(n)$ then $\|w\|_\infty\le \|w\|_{L^\infty(S_\tau)}\cdot e^{cn(1-\tau)^2}$.
\item \label{3rd} If $q\in\Pol(n)$ then
$\|w\|_\infty\le C\sqrt{n}\|w\|_2$.
\end{enumerate}
\end{lem}

\begin{proof} \ref{1st}: By assumption, $1-\tau_0\le \tau\le 1-a/\sqrt{n}$, where we write
$\tau=\tau(j)=j/n.$
It follows that $\dist(\Gamma_\tau,\Gamma)\ge c/\sqrt{n}$ for some
$c=c(a)>0$ by Lemma \ref{movin}.
By Lemma \ref{propp}
there is a constant $C=C(c)$ such that $|w(\zeta)|\le C\sqrt{n}\|w\|_{2}$ when $\zeta\in S_\tau$.

Now consider the subharmonic function
\begin{equation*}\label{smock}u_n(\zeta)=\frac 1 n\log\left(\frac {|q(\zeta)|^2}{Cn\|w\|_2^2}\right).\end{equation*}
Note that $u_n(\zeta)\le \tau\log|\zeta|^2+O(1)$ as $\zeta\to\infty$ and $u_n\le Q$ on $S_\tau$, and $u_n$ is subharmonic on $\C$. Hence $u_n\le\eqpot_\tau$ on $\C$, which is
equivalent to (i).

\ref{2nd}: Let $U_\tau$ be the unbounded component of $\hat{\C}\setminus S_\tau$.
The function $$v_n(\zeta):=\frac 1 n\log|q(\zeta)|^2$$ is by assumption $\le Q$ on $S_\tau$ and grows no faster than $\log|\zeta|^2+O(1)$ as $\zeta\to\infty$.
It follows that the function $$s_n(\zeta)=v_n(\zeta)-\eqpot_\tau(\zeta)-(1-\tau)G_\tau(\zeta)$$ is subharmonic and bounded in $U_\tau$, where $G_\tau(\zeta)$ is Green's function
of $U_\tau$ with pole at infinity. Moreover, $s_n\le 0$ on $\Gamma_\tau$ and $\limsup_{\zeta\to\infty} s_n(\zeta)\le 0$.
By a suitable version of the maximum principle, $s_n\le 0$ on $U_\tau$.
Hence if $\zeta\in S\setminus S_\tau$,
$$v_n(\zeta)\le \eqpot_\tau(\zeta)+(1-\tau)G_\tau(\zeta)\le Q(\zeta)+c(1-\tau)^2,$$
where the estimate $G_\tau(\zeta)\le c(1-\tau)$ was obtained using that the distance from $\zeta$ to $\Gamma_\tau$ is $\le c_2(1-\tau)$ by Lemma \ref{movin} and the $G_\tau$ continues harmonically across $\Gamma_\tau$.

We have shown that $|w|^2\le|q|^2e^{-nQ^S} e^{nc(1-\tau)^2}$ on $\C$, as desired.

\ref{3rd} Put $\tau=1-1/\sqrt{n}$ and choose $c>0$ so that $\dist(\Gamma_\tau,\Gamma)\ge c/\sqrt{n}$. By Lemma \ref{propp}, there is $C$ such that
$|w(\zeta)|\le C\sqrt{n}\|w\|_2$ for $\zeta\in S_\tau$. The result now follows from \ref{2nd}.
\end{proof}

\begin{cor} \label{goodb} There is a constant $C$ such that $\bfR_n\le Cn$ on $\C$. In particular, the rescaled $1$-point functions $R_n$
are uniformly bounded on $\C$.
\end{cor}

\begin{proof} Let $\bfk_n$ be the reproducing kernel of the space of holomorphic polynomials of degree at most $n-1$ with the norm of
$L^2(e^{-nQ^S})$. Now fix $\zeta\in S$ and put $q(\eta)=\bfk_n(\eta,\zeta)/\sqrt{\bfk_n(\zeta,\zeta)}$ and $w=qe^{-nQ^S/2}$. Then $\|w\|_2=1$, and so $\|w\|_\infty\le C\sqrt{n}$ by Lemma \ref{l6}. Finally,
$\bfR_n(\zeta)=|w(\zeta)|^2.$
\end{proof}

\subsection{Discarding lower order terms}\label{disc}
Let $N_\Gamma$ be the belt \eqref{belt}. Below we fix an arbitrary point
$\zeta\in N_\Gamma.$
Also fix a number $\tau_0<1$ such that the curves $\Gamma_\tau$ are regular for all $\tau$ with $\tau_0\le \tau\le 1$. Now write
$$\tau=\tau(j)=j/n,\qquad (j\le n-1).$$
By Lemma \ref{l6}, there is a number $C$ such that
$$\tau\le \tau_0\quad  \Rightarrow\quad |w_{j,n}(\zeta)|^2\le
Cne^{-n(Q^S-\check{Q}_{\tau_0})(\zeta)}\le Cne^{-cn},$$
where $c=\inf_{N_\Gamma}\{Q-\check{Q}_{\tau_0}\}>0$. (We shall elaborate on the constant $c$ below.)
Hence $$\bfR_n(\zeta)\sim\sum_{j=\tau_0 n}^{n-1}|w_{j,n}(\zeta)|^2.$$

Next fix $j$ such that $\tau_0 n\le j\le n-\sqrt{n}\log n$, i.e.,
$\tau_0\le \tau\le 1-\delta_n$ where we write $\delta_n=\log n/\sqrt{n}$.

We will denote by $V_\tau$ the harmonic continuation of the harmonic function $\check{Q}_\tau$ on $U_\tau$ across the analytic curve $\Gamma_\tau$. Considering the growth as $\zeta\to\infty$, we obtain the basic identity
\begin{equation}\label{good}V_\tau(\zeta)=\re\calQ_\tau+\tau\log|\phi_\tau(\zeta)|^2,
\end{equation}
where $\calQ_\tau$ is the holomorphic function on $U_\tau$ with $\re \calQ_\tau=Q$ on $\Gamma_\tau$ and $\im \calQ_\tau(\infty)=0$.

The function $Q-V_\tau$, considered in a neighbourhood of the curve $\Gamma_\tau$, plays a central role for the theory.
Note that $Q-V_\tau$ vanishes to first order at $\Gamma_\tau$, and increases quadratically with the distance from $\Gamma_\tau$,
\begin{equation}\label{ridge}(Q-V_\tau)(\zeta)= 2\Lap Q(\zeta)\cdot \dist(\zeta,\Gamma_\tau)^2+\cdots.\end{equation}
The function $Q-V_\tau$ might be called a "parabolic ridge" about the curve $\Gamma_\tau$.
(The proof of \eqref{ridge} is just a case of using Taylor's formula, see e.g. \cite[Section 5.3]{AKM} for details.)

In particular, it follows from \eqref{ridge} that
there is a number $c>0$ such that, if $\zeta\in N_\Gamma$,
\begin{equation}\label{skurk}(Q-\check{Q}_\tau)(\zeta)=(Q-V_\tau)(\zeta)\ge c\dist(\zeta,S_\tau)^2.\end{equation}

By Lemma \ref{movin}, we have $\dist(\Gamma_\tau,\Gamma)\ge c(1-\tau)\ge c\delta_n$ for some constant $c>0$.
By \eqref{skurk},
$n(Q-\check{Q}_\tau)(\zeta)\ge c\log^2 n$, so
$$|w_{j,n}(\zeta)|^2\le Cne^{-n(Q-\check{Q}_\tau)(\zeta)}\le Cne^{-c\log^2 n},\qquad (j\le n-\sqrt{n}\log n).$$
We have shown that
$$\bfR_n(\zeta)\sim \sum_{j=n-\sqrt{n}\log n}^{n-1}|w_{j,n}(\zeta)|^2, \quad (n\to\infty)$$
in the sense that the difference of the left and right sides converges uniformly to zero on $N_\Gamma$, as $n\to\infty$.

\section{Foliation flow and quasipolynomials}
\subsection{Basic definitions}
Fix a large integer $n$. For an integer $j$ in the range
$n-\sqrt{n}\log n\le j\le n-1$ we put
$\tau=\tau(j)=j/n$.
Thus $1-\delta_n\le \tau< 1.$

When $n$ is large, $\Gamma_\tau$ is a real-analytic Jordan curve and thus
the map $\phi_\tau:U_\tau\to \D_e$ can be continued analytically across $\Gamma_\tau$.

Let $T=T_\epsilon=\{\zeta;\,\dist(\zeta,\Gamma)<\epsilon\}$ be an arbitrarily small but fixed tubular neighbourhood of $\Gamma$.

For a real parameter $t$ with $|t-\tau|\le 2\delta_n$, we denote by $L_{\tau,t}$ the level set
$$L_{\tau,t}=\{\zeta\in T;\, (Q-V_\tau)(\zeta)= t^2\}.$$
Of course $L_{\tau,0}=\Gamma_\tau$.

By the relation \eqref{ridge}, we see that for $t\ne 0$, $L_{\tau,t}$ is the disjoint union of two analytic Jordan curves,
$L_{\tau,t}=\Gamma_{\tau,t}^-\cup\Gamma_{\tau,t}^+$ where $\Gamma_{\tau,t}^-\subset\Int \Gamma_\tau$ and $\Gamma_{\tau,t}^+\subset\Ext\Gamma_\tau$. Define
$\Gamma_{\tau,t}=\Gamma_{\tau,t}^-$ if $-\delta_n\le t\le 0$ and $\Gamma_{\tau,t}=\Gamma_{\tau,t}^+$ if $0\le t\le \delta_n$. Then $\Gamma_{\tau,t}$ is
an analytic Jordan curve depending analytically on the parameter $t$; see
\cite[Section 3]{HW} for further details.

Let $U_{\tau,t}$ be the exterior domain of $\Gamma_{\tau,t}$ and consider the simply connected domain $\phi_\tau(U_{\tau,t})\subset \hat{\C}$. This is a slight perturbation of the exterior disk $\D_e$.

We will denote by
$$\psi_t:\D_e\to \phi_\tau(U_{\tau,t})$$
the normalized univalent mapping between the indicated domains. Note that $\psi_t$ continues analytically across $\T$ and obeys the basic relation
\begin{equation*}\label{implicit}(Q-V_\tau)\circ\phi_\tau^{-1}\circ
\psi_t\equiv t^2,\qquad \text{on}\quad \T.
\end{equation*}
By the "foliation flow emanating from $\Gamma_\tau$", we mean the family of Jordan curves
$$\Gamma_{\tau,t}=\phi_\tau^{-1}\circ\psi_t(\T),\qquad (-2\delta_n\le t\le 2\delta_n).$$

Given a point $0\in \Gamma$ and a number $\tau$ as above, it will be important to have a good estimate of the "local stopping time" $t_\infty$ when the foliation flow emanating from $\Gamma_\tau$ hits the point $0$.

To make this precise, recall that we have placed the coordinate system so that the positive axis points in the outer normal direction with respect to $S$ at the point $0\in\Gamma$. We denote by
$x_{j,n}$ the closest point to $0$ on $\Gamma_\tau$, on the negative half axis.

Let $t_\infty>0$ be the smallest value of the flow parameter $t$ such that the curve
$\phi_\tau^{-1}\circ\psi_t(\T)$ hits the point $0\in\Gamma_1$,
$$0\in \phi_\tau^{-1}\circ\psi_{t_\infty}(\T).$$

\begin{lem} Suppose that $\tau=j/n$ satisfies
$1>\tau\ge 1-\log n/\sqrt{n}$, and write
$$\gone(\zeta):=\frac {|\phi_1'(\zeta)|}{\sqrt{\Lap Q(\zeta)}},\qquad (\zeta\in\Gamma_1).$$ Then
$$t_\infty=\frac {\gone(0) (1-\tau)}{\sqrt{2}}+O((1-\tau)^2),\qquad (n\to\infty).$$
\end{lem}

\begin{proof}
It follows by Taylor's formula and Lemma \ref{movin} that
$$(Q-V_\tau)(0)\sim 2\Lap Q(x_{j,n})x_{j,n}^2\sim(1-\tau)^2\frac {|\phi_\tau'(x_{j,n})|^2}
{2\Lap Q(x_{j,n})}.$$
Moreover if $\eta\in\T$ is such that $\phi_\tau^{-1}\circ\psi_{t_\infty}(\eta)=0$ then
$$(Q-V_\tau)\circ \phi_\tau^{-1}\circ\psi_{t_\infty}(\eta)= t_\infty^2,$$
so we obtain
$$t_\infty^2 \sim (1-\tau)^2\frac {|\phi_\tau'(x_{j,n})|^2}
{2\Lap Q(x_{j,n})}.$$
The lemma follows by taking square-roots.
\end{proof}

By the same token, if we define
\begin{equation*}\label{gzdeff}\gtau(\zeta)=\frac {|\phi_\tau'(\zeta)|}{\sqrt{\Lap Q(\zeta)}},\qquad \zeta\in\Gamma_\tau,\end{equation*}
then the flow emanating from the point $\zeta\in \Gamma_\tau$ will hit $\Gamma=\Gamma_1$ after time $t_\infty(\zeta)$ where
\begin{equation*}\label{tinef}t_\infty(\zeta)=\frac {\gtau(\zeta)(1-\tau)}{\sqrt{2}}+O((1-\tau)^2).\end{equation*}

In symbols, given a point $\zeta\in \Gamma_\tau$, this \textit{stopping time} $t_\infty(\zeta)$ is defined to be the smallest number $t>0$ such that
$$\phi_\tau^{-1}\circ\psi_{t}\circ\phi_\tau(\zeta)\in\Gamma.$$

\subsection{Quasipolynomials}\label{qps}

Fix numbers $\tau_0<1$ and $\epsilon>0$, and an integer $n_0$, such that when $n\ge n_0$ and $\tau:=j/n$ satisfies $1-2\delta_n<\tau<1$,
then $\phi_\tau^{-1}$ continues analytically to a univalent mapping defined on the exterior disk $\D_e(\tau_0-\epsilon)$.

Now fix $\tau=j/n$ with $1-2\delta_n<\tau<1$ and
retain the notation $\phi_\tau^{-1}$ for the extended map.
We define two sets $K$ and $X$ with $K\subset X\subset\Pc S$ by
\begin{equation*}\label{sets}K:=\C\setminus \phi_\tau^{-1}(\D_e(\tau_0-\epsilon)),\qquad
X:=\C\setminus \phi_\tau^{-1}(\D_e(\tau_0)).\end{equation*}

Similar to \cite{HW}, we define a "hard edge quasipolynomial" $F_{j,n}$ on $\C\setminus K$ by
\begin{equation}\label{quasi}F_{j,n}=n^{1/4}\cdot \phi_\tau'\cdot \phi_\tau^j\cdot e^{n\calQ_\tau/2}\cdot f_\tau\circ\phi_\tau,\qquad (\tau=j/n)\end{equation}
where
$$\qquad f_\tau:=\pi^{-1/4}(\phi_\tau')^{-1/2}e^{\calH_\tau\circ\phi_\tau^{-1}/2}.$$

Here the functions $\calQ_\tau$ and $\calH_\tau$ are holomorphic on $\hat{\C}\setminus K$ and satisfy
\begin{equation}\label{ndefs}\re \calQ_\tau=Q\qquad
\text{and}\qquad
\re\calH_\tau=\log \sqrt{\Lap Q}-\log\fif_{j,n}\qquad \text{on}\quad  \Gamma_\tau,\end{equation}
where $\fif_{j,n}$ is defined on $\Gamma_\tau$ by
\begin{equation*}\label{xjn}\fif_{j,n}(\zeta):=\fif[\xi_{j,n}\cdot\gtau
(\zeta)],\qquad \xi_{j,n}:=\frac {j-n}{\sqrt{n}}, \quad \zeta\in\Gamma_\tau,\end{equation*}
where as always $\fif$ is the free boundary function \eqref{ffun}.

Since the functions $Q$, $\log\sqrt{\Lap Q}$, and $\log \fif_{j,n}$ are all real-analytic on $\Gamma_\tau$, we can indeed find holomorphic functions $\calQ_\tau$ and $\calH_\tau$ having these
properties, provided that $\tau_0<1$ and $\delta>0$ are chosen appropriately. We fix $\calQ_\tau$ and $\calH_\tau$ uniquely by requiring that their imaginary parts vanish at $\infty$.

Note that $F_{j,n}$ is holomorphic on $\C\setminus K$. Moreover, we can write
$$F_{j,n}=n^{1/4}\Lambda_{j,n}[f_\tau],$$ where
$\Lambda_{j,n}$, the "positioning operator", is defined by
$$\Lambda_{j,n}[f]=\phi_\tau'\cdot \phi_\tau^j\cdot e^{n\calQ_\tau/2}\cdot f\circ\phi_\tau.$$
It is easy to show that $\Lambda_{j,n}$ has the following isometry property: for all suitable functions $f,g$,
\begin{equation}\label{isometry}\int_{\C\setminus K}\Lambda_{j,n}[f]\overline{\Lambda_{j,n}[g]}\,
e^{-nQ}
\, dA=\int_{\D_e(\tau_0-\delta)}f\bar{g}\, e^{-n(Q-V_\tau)\circ\phi_\tau^{-1}}
\, dA.\end{equation}
Cf. \cite{HW}, Section 3.

We observe for later use that, by \eqref{good}, \eqref{ridge} and the definition of $F_{j,n}$,
\begin{equation}\label{gest}
|F_{j,n}|^2e^{-nQ}=\sqrt{n}|\phi_\tau'||f_\tau\circ\phi_\tau|^2e^{-n(Q-V_\tau)}
\le C\sqrt{n}e^{-cnd_\tau^2},\qquad (d_\tau(\zeta):=\dist(\zeta,\Gamma_\tau))
\end{equation}
on $\C\setminus K$, where $c$ is a positive constant.

\begin{rem} If in \eqref{ndefs} we redefine the function $\calH_\tau$ by $\re \calH_\tau=\log\sqrt{\Lap Q}$ we obtain precisely
the quasipolynomials used in the free boundary case in \cite{HW}.
\end{rem}

\subsection{Integration of quasipolynomials}
For each point $\zeta=\phi_\tau^{-1}(\eta)\in\Gamma_\tau$ we use the stopping time
 $t_\infty(\zeta)\sim \gtau(\zeta)\cdot (1-\tau)/\sqrt{2}$ as follows.

 For $j,n$ with $n-\sqrt{n}\log n< j<n$ we consider the two closely related domains
\begin{equation*}\label{trunc}\begin{split}D_{j,n}&=\{\psi_t(\eta)\in\C;\quad \eta\in\T, \quad -2\delta_n\le t\le t_\infty(\phi_\tau^{-1}(\eta))\},\\
\tilde{D}_{j,n}&=\{(\eta,t)\in\T\times\R;\qquad\quad\, -2\delta_n\le t\le t_\infty(\phi_\tau^{-1}(\eta))\}.\\
\end{split}
\end{equation*}
Here, as always, $\tau=j/n$.

Note that the definition of $D_{j,n}$ is set up so that the outer boundary of $\phi_\tau^{-1}(D_{j,n})$ coincides with the outer boundary of $S$,
$$\Pc \phi_{\tau}^{-1}(D_{j,n})=\Pc S.$$

Following \cite{HW} we define the
flow map $\Psi$ by
\begin{equation*}\label{fflow}\Psi:\tilde{D}_{j,n}\to D_{j,n},\qquad \Psi(t,\eta)=\psi_t(\eta).\end{equation*}

A consideration of the Jacobian of $\Psi$ shows that
if $f$ is an integrable function on $D_{j,n}$, then
\begin{equation}\label{cov0}\int_{D_{j,n}} f\, dA=2\int_{\tilde{D}_{j,n}}  f\circ\psi_t \cdot
\re(-\bar{\eta}\cdot \d_t\psi_t(\eta)\cdot\overline{\psi_t'(\eta)})\, dt \,ds(\eta).\end{equation}

\smallskip

We shall need the following lemma, which is obtained by a straightforward adaptation of the main lemma in \cite[section 4]{HW}.

\begin{lem} \label{hwlem} In the above situation, we have for $\eta\in\T$ that (as $n\to\infty$)
\begin{equation*}\label{cov}
\begin{split}
2\sqrt{n}|f_\tau\circ\psi_t(\eta)|^2&e^{-n(Q-V_\tau)\circ\psi_t(\eta)}\re(-\bar{\eta}\cdot \d_t\psi_t(\eta)\cdot\overline{\psi_t'(\eta)})
=\sqrt{\frac n {\pi}}\cdot \frac {e^{-nt^2}}
{\fif_{j,n}\circ \phi_\tau^{-1}(\eta)}
\cdot (1+O(t)).\\
\end{split}
\end{equation*}
\end{lem}

\begin{proof}[Remark on the proof] This follows from the relation (4.1.3) in \cite[Lemma 4.1.2]{HW}. To translate between the settings we observe that our present "$Q$" is denoted "$2Q$" in \cite{HW}, our present $f_\tau$ equals
to the function denoted $f_{j,n}^{\langle0\rangle}$ in \cite{HW}, \textit{multiplied} with the function $e^{\calK_\tau\circ\phi_\tau^{-1}/2}$ where $\re\calK_\tau=-\log\fif_{j,n}$ on $\Gamma_\tau$,
and finally $\psi_t(\eta)=\eta+O(t)$ for all $\eta\in\T$.
\end{proof}

Using the isometry property \eqref{isometry} and
the asymptotic in Lemma \ref{hwlem}, we obtain that
\begin{equation}\label{norm2}\begin{split}\int_{\phi_\tau^{-1}(D_{j,n})}
&|F_{j,n}|^2 e^{-nQ}\, dA=\sqrt{n}\int_{D_{j,n}}|f_\tau|^2e^{-n(Q-V_\tau)\circ\phi_\tau^{-1}}\, dA\\
&=\sqrt{n}\int_{\tilde{D}_{j,n}}
|f_\tau\circ\psi_t|^2 e^{-n(Q-V_\tau)\circ\phi_\tau^{-1}\circ \psi_t}\,
\cdot \re(-\bar{\eta}\cdot \d_t\psi_t(\eta)\cdot\overline{\psi_t'(\eta)})\, dt \,ds(\eta)\\
&= \int_\T \frac {ds(\eta)} {\fif(\gtau(\phi_\tau^{-1}(\eta))\xi_{j,n})}\,\frac {\sqrt{n}}{\sqrt{\pi}}\int_{-2\delta_n}^{t_\infty(\phi_\tau^{-1}(\eta))}e^{-nt^2}(1+O(t))\, dt\\
&=\int_\T (1+O(1/\sqrt{n})) ds(\eta)=1+O(1/\sqrt{n}).\end{split}
\end{equation}
Between the second and third line, we replaced $\fii_{j,n}\circ\psi_t$ by $\fii_{j,n}$; the associated error is absorbed it the $O(t)$-term.
Between the third and fourth line, we used the asymptotic
$t_\infty(\phi_\tau^{-1}(\eta))\sim \ell_\tau(\phi_\tau^{-1}(\eta))\cdot (1-\tau)/\sqrt{2}$ in the upper $t$-integration limit.

\smallskip

We now fix, once and for all, a smooth function $\chi_0$ with $\chi_0=0$ on $K$ and $\chi_0=1$ on $\C\setminus X$.

\begin{lem} \label{11} There is a constant $c>0$ such when $n-\sqrt{n}\log n<j<n$,
\begin{equation*}\label{norm3}\int_{S}  |\chi_0\cdot F_{j,n}|^2e^{-nQ}\, dA=1+O(1/\sqrt{n}),\qquad (n\to\infty).\end{equation*}
\end{lem}

\begin{proof} Consider with $\tau=j/n$ the contribution
\begin{align*}\int_{S\setminus \phi_\tau^{-1}(D_{j,n})}\chi_0^2\cdot |F_{j,n}|^2e^{-nQ}\, dA
=n^{1/2}
\int_{\phi_\tau(S)\setminus D_{j,n}}\tilde{\chi}_0^2\cdot |f_\tau|^2
\cdot e^{-n(Q-V_\tau)\circ\phi_\tau^{-1}}\, dA.
\end{align*}
Here $\tilde{\chi}_0$ is given by $\tilde{\chi}_0=\chi_0\circ\phi_\tau$ wherever $\phi_\tau$ is defined, and $\tilde{\chi}_0=0$ elsewhere.

Since $f_\tau$ is bounded on $\phi_\tau(S\setminus K)$, we see from \eqref{gest} that
\begin{equation}\label{fex}\begin{split}\int_{S\setminus \phi_\tau^{-1}(D_{j,n})}&\chi_0^2\cdot |F_{j,n}|^2e^{-nQ}\, dA\le C\sqrt{n}
\int_{S\setminus \phi_\tau^{-1}(D_{j,n})}e^{-n(Q-V_\tau)}\, dA\le C\sqrt{n}e^{-c\log^2 n}.\\
\end{split}
\end{equation}
The statement thus follows from \eqref{norm2}.
\end{proof}

We now define a weighted quasipolynomial $w_{j,n}^\sharp$ by
\begin{equation}\label{app2}w_{j,n}^\sharp=\chi_0\cdot F_{j,n}\cdot e^{-nQ^S/2}.\end{equation}
(It is understood
that $w_{j,n}^\sharp=0$ on $K$.)

\begin{cor} \label{bom} For $n-\sqrt{n}\log n<j<n$, we have that $\|w_{j,n}^\sharp\|_2=1+O(1/\sqrt{n})$.
\end{cor}

We shall see that, in the given range of $j$'s, the function $w_{j,n}^\sharp$ is a good approximation to the $j$:th weighted
orthonormal polynomial $w_{j,n}=p_{j,n}e^{-nQ^S/2}$, at least if we assume radial symmetry of the potential. A few boundary profiles of probability densities $|w_{j,n}|^2$  are
depicted in Figure 4.

\begin{figure}[ht]\label{m_graph}
\begin{center}

\includegraphics[width=.3\textwidth]{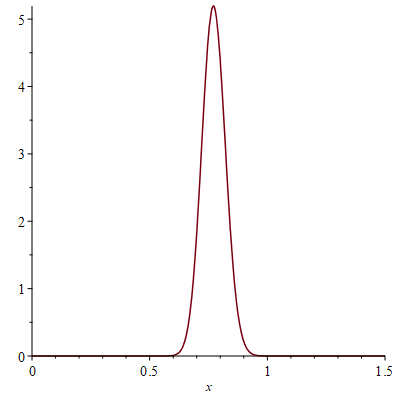}
\hspace{.03\textwidth}
\includegraphics[width=.3\textwidth]{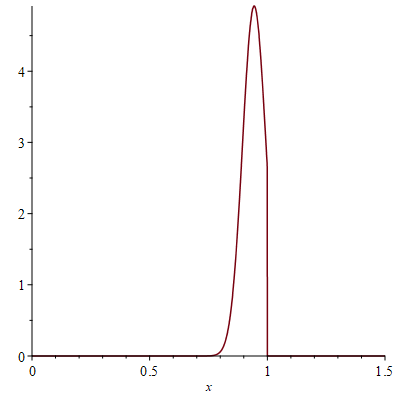}
\hspace{.03\textwidth}
\includegraphics[width=.3\textwidth]{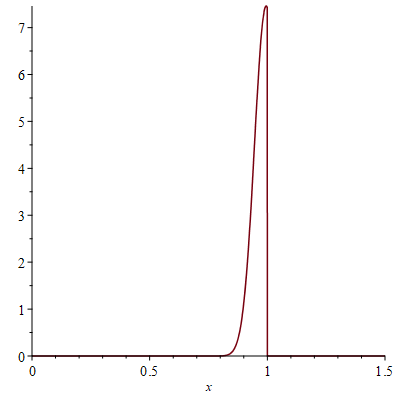}
\end{center}
\caption{Boundary profiles of squared weighted orthonormal polynomials $|w_{j,n}|^2$ with $n=101$ and $j=60,90,100$, respectively. (Here $\gone(0)=1$.)}
\end{figure}

\subsection{Approximate orthogonality} \label{shit} In the spirit of \cite[Section 4.8]{HW}, we now prove that the weighted quasipolynomial $w_{j,n}^\sharp$ of \eqref{app2} has the following "approximate orthogonal property".

\begin{lem}\label{appop} Assume that $Q$ is radially symmetric.
Suppose that $n-\sqrt{n}\log n\le j<n$ and pick a weighted polynomial $w=pe^{-nQ^S/2}$ with $\degree(p)\le j-1$. Then
we have the estimate
$$\left|\int_{S} w\bar{w}_{j,n}^\sharp\, dA\right|\le Cn^{-1/2}\|w\|_2.$$
\end{lem}

 \begin{proof} We fix a holomorphic polynomial $p$ of degree at most $j-1$, and
define a holomorphic function $g_\tau$ by $\Lambda_{j,n}[g_\tau]=p$. A computation gives
$$g_\tau(\eta)=\eta^{-j}\cdot
[\phi_\tau^{-1}]'(\eta)
e^{-n\calQ_\tau\circ\phi_\tau^{-1}(\eta)/2}\cdot (p\circ\phi_\tau^{-1})(\eta),\qquad (\tau=j/n).$$
Observe that $g_\tau(\infty)=0$.

We must estimate the integral
\begin{equation}\label{inte}\int_{S} \chi_0 \cdot p\cdot \bar{F}_{j,n}\, e^{-nQ}= n^{1/4}\int_{\phi_\tau(S\setminus K)} \tilde{\chi}_0 \cdot g_\tau\cdot\bar{f}_\tau\, e^{-n(Q-V_\tau)\circ\phi_\tau^{-1}/2}.\end{equation}

By \eqref{fex} and the Cauchy-Schwarz inequality we have, for small enough $c>0$,
\begin{align*}\left|\int_{S\setminus \phi_\tau^{-1}(D_{j,n})}\chi_0 \cdot p\cdot \bar{F}_{j,n}\, e^{-nQ}\,\right|
&
\le Ce^{-c\log^2 n}\|w\|_2.
\end{align*}

It hence suffices to estimate the integral
\begin{equation}\label{gesto}\int_{D_{j,n}}g_\tau\bar{f}_\tau e^{-n(Q-V_\tau)\circ\phi_\tau^{-1}}
=\int_{D_{j,n}}h_\tau\cdot |f_\tau|^2\, e^{-n(Q-V_\tau)\circ\phi_\tau^{-1}}.
\end{equation}
where we write
$h_\tau=g_\tau/f_\tau.$
Note that $h_\tau$
is holomorphic in a neighbourhood of $\overline{\D}_e$ and vanishes at $\infty$ (because $f_\tau(\infty)>0$).

Note that Lemma \ref{hwlem} implies that the last integral in \eqref{gesto} can be written
$$\int_{\tilde{D}_{j,n}}\frac {h_\tau\circ \psi_t(\eta)}{\fii(\xi_{j,n}\gtau(\phi_\tau^{-1}(\eta)))}\cdot e^{-nt^2}(1+O(t))\, ds(\eta)\, dt=I_1+I_2,$$
where $I_1$ is formed by picking "$1$" in the parenthesis $(1+O(t))$ and $I_2$ corresponds to "$O(t)$".

The integral $I_2$ is easily estimated by the following argument.
Since $1/f_\tau$ is bounded on $D_{j,n}$ we have $|h_\tau|\le C|g_\tau|$, and so
\begin{align*}|I_2|\le C\int_{\tilde{D}_{j,n}}|t|e^{-nt^2}|g_\tau\circ\psi_t(\eta)|\, dt ds(\eta).\end{align*}
By the Cauchy-Schwarz inequality and the isometry property, we see that
\begin{align*}
|I_2| &\le C(\int_{D_{j,n}}|g_\tau|^2e^{-n(Q-V_\tau)\circ\phi_\tau^{-1}}\, dA)^{1/2}(\int_{\R}t^2e^{-nt^2}\, dt)^{1/2}\\
&\le Cn^{-3/4}(\int_{\phi_\tau^{-1}(D_{j,n})}|p|^2e^{-nQ}\, dA)^{1/2}\le Cn^{-3/4}\|w\|_2.
\end{align*}

We must finally estimate the integral
$$I_1=\int_{\tilde{D}_{j,n}}\frac {h_\tau\circ \psi_t(\eta)}{\fii(\xi_{j,n}\gtau(\phi_\tau^{-1}(\eta)))}\cdot e^{-nt^2}\, ds(\eta)\, dt.$$
The estimation of $I_1$ is in general nontrivial, but becomes very simple if the stopping time $t_\infty$ is uniform on $\Gamma_\tau$, i.e., if
\begin{equation}\label{unif}\tilde{D}_{j,n}=\T\times (-2\delta_n,t_\infty)\end{equation}
where $t_\infty>0$ is a constant. It is easily seen that \eqref{unif} holds for radially symmetric potentials.

If \eqref{unif} is satisfied, then by the mean-value property of the function $h_\tau\circ\psi_t$,
$$I_1= \frac 1 {\sqrt{\pi}\fii(\xi_{j,n}k)}\int_{-2\delta_n}^{t_\infty}e^{-nt^2}\, dt\int_\T h\circ\psi_t(\eta)\, ds(\eta)=0,$$
where we wrote $k$ for the constant value of $\gtau\circ\phi_\tau^{-1}$ on $\T$.
The proof is completed by noting that the inner product in \eqref{inte} is dominated by $Cn^{1/4}(|I_1|+|I_2|)$.
\end{proof}

\begin{rem} Our assumption of uniform stopping time (or radial symmetry of the potential) is used only for the estimation of the integral $I_1$ above.
\end{rem}

\section{The hard edge scaling limit} \label{oopp}
We will now prove Theorem \ref{mthm2}. We shall start by constructing a presumptive local approximation of the 1-point function.

\subsection{The approximate 1-point function} \label{central}
By the definition of $F_{j,n}$ in \eqref{quasi}, and the basic identity \eqref{good}, we have on the set $\phi_\tau^{-1}(D_{j,n})$,
\begin{equation*}\label{gumb}|w_{j,n}^\sharp|^2\sim\sqrt{\frac n \pi}\cdot |\phi_\tau'|\cdot e^{-n(Q-V_\tau)}e^{\re \calH_\tau}.\end{equation*}

We shall estimate, for $\zeta$ close to the point $0\in\Gamma_1$, the sum
$$\bfR_n^\sharp(\zeta):=\sum_{j=n-\sqrt{n}\log n}^{n-1}|w_{j,n}^\sharp(\zeta)|^2,$$
which we shall call "approximate 1-point function".

More precisely, rescale about $0\in\d S$ by
\begin{equation}\label{rescale}\zeta=\frac z{\sqrt{n\Lap Q(0)}},\qquad R_n^\sharp(z)=\frac 1 {n\Lap Q(0)}\bfR_n^\sharp(\zeta).\end{equation}

 We suppose in the following that $z$ stays in a compact subset of $\L$; we will at first assume that $z$ is real and negative. Observe in this case that
 \begin{equation}\label{gumb2}|w_{j,n}^\sharp(\zeta)|^2\sim\sqrt{\frac n \pi}\frac {|\phi_1'(0)|} {\fif\circ(\xi_{j,n}\cdot \gz)}e^{-n(Q-V_\tau)(\zeta)}\sqrt{\Lap Q(0)},\end{equation}
where we write
$$\gz:=\gone(0)= \frac{\phi_1'(0)}{\sqrt{\Lap Q(0)}}.$$

We apply Taylor's formula about the closest point $x_{j,n}\in\Gamma_\tau\cap \R$ to $0$, computed in Lemma \ref{movin},
$$x_{j,n}=-\frac {|\phi_1'(0)|}{2\Lap Q(0)}(1-j/n)+\cdots.$$
Note first that
$$\sqrt{n\Lap Q(0)}x_{j,n}=-\frac {\gz } 2 \sqrt{n}(1-j/n)\sim\frac \gz 2 \cdot \xi_{j,n}$$
where (as before) $\xi_{j,n}=(j-n)/\sqrt{n}$.

Taylor's formula gives
\begin{align*}n(Q-V_\tau)(\zeta)&=n(Q-V_\tau)\left(x_{j,n}+\frac {z+
\frac {\gz } 2\sqrt{n}(1-j/n)}
{\sqrt{n\Lap Q(0)}}\right)\\
&\sim 2\Lap Q(x_{j,n})\left(\frac {z+
\frac {\gz } 2\sqrt{n}(1-j/n)}
{\sqrt{\Lap Q(0)}}\right)^2\\
&\sim \frac 1 2 (2z+\gz \sqrt{n}(1-j/n))^2.
\end{align*}
Hence, by \eqref{gumb2},
$$\frac 1 {n\Lap Q(0)}|w_{j,n}^\sharp(\zeta)|^2\sim
\frac 1 {\sqrt{2\pi}}\frac {\gz }{\sqrt{n}} \frac 1 {\fif(\gz (j-n)/\sqrt{n})} e^{-\frac 1 2(2z+\gz \sqrt{n}(1-j/n))^2}.$$

Setting $k=n-j$ and summing in $k$, it follows that
\begin{equation}\label{sles}\begin{split}R_n^\sharp(z)
&\sim\frac 1 {\sqrt{2\pi}}\sum_{k=1}^{\sqrt{n}\log n}
\frac {\gz }{\sqrt{n}}\frac {\gamma(2z+k\gz /\sqrt{n})}{\fif(-\gz k/\sqrt{n})}\\
&\sim\int_0^\infty \frac {\gamma(2z+t)}{\fif(-t)}\, dt=H(2z).\\
\end{split}
\end{equation}
Here we replaced a Riemann sum with an integral and used the definition of the $H$-function \eqref{hfun}.

If $z\in\L$ is not real, we merely select a new ($n$-dependent) coordinate system so that the origin corresponds to the point $\zeta=z/\sqrt{n\Lap Q(0)}$, and so that
the outwards unit normal to $\Gamma$ points in the positive real direction at $\zeta$.
Repeating the above argument, we see that $R_n^\sharp(z)\to H(z+\bar{z})$.
It is easily seen that the convergence is bounded on $\C$ and
locally uniform on $\C\setminus(i\R)$.

\subsection{Quantization of the quasipolynomials and error estimates} To complete our analysis, we must show that $R_n^\sharp$ well approximates the rescaled 1-point function $R_n$ on bounded subsets of the $z$-plane, provided that we assume radial symmetry.
To do this, we shall, following \cite[Section 4.9]{HW}, correct our quasipolynomials to actual polynomials, by using $L^2$-estimates for the norm minimal solution
to suitable $\dbar$-problems.

Fix $j$ such that $\tau=j/n$ is in the range $\delta_n<\tau<1$.

We correct $F_{j,n}$ to a polynomial $p_{j,n}^*$ of degree at most $j$ in the following way. Let $u_0$ be the
$L^2(e^{-nQ})$-minimal solution to the $\dbar$-problem
$$\dbar u=\dbar\chi_0\cdot F_{j,n},\qquad \text{and}\qquad u(\zeta)=O(\zeta^j),\qquad (j\to\infty).$$
We set $p_{j,n}^*=\chi_0\cdot F_{j,n}-u_0$ and observe that $p_{j,n}^*$ is an entire function which is $O(\zeta^j)$ as $\zeta\to\infty$, i.e., it is a polynomial of degree at most $j$.

Let us now briefly recall how to derive the basic estimate
\begin{equation}\label{hoe}\int|u_0|^2e^{-nQ}\le \frac C n \int |\dbar u_0|^2e^{-nQ}.\end{equation}
To do this, we fix a constant $\alpha$ with $1<\alpha<2$ and use the modification of $\check{Q}_\tau$ defined by
$$\phi_n(\zeta)=\check{Q}_\tau(\zeta)+\frac \alpha n\log(1+|\zeta|^2),$$
which is strictly subharmonic and $C^{1,1}$-smooth in $\C$. Recall that $\check{Q}_\tau=Q$ on $S_\tau$, and so on the support of $\dbar \chi_0$. It thus follows from Hörmander's well-known estimate in \cite[Section 4.2]{H} that
\begin{equation}\label{hoe2}\int|u_0|^2e^{-n\phi_n}\le \int_{S_\tau}|\dbar u_0|^2\frac {e^{-n\phi_n}}{n\Lap \phi_n}.\end{equation}
Since
$-nQ\le -n\phi_n+\const$ on $\C$,
we conclude the estimate in \eqref{hoe}.

Combining \eqref{hoe} with the estimate \eqref{gest}, we conclude that
\begin{equation}\label{ULK}\int_\C|p_{j,n}^*-\chi_0\cdot F_{j,n}|^2e^{-nQ}\, dA\le\frac C n\int_\C |\dbar \chi_0 \cdot F_{j,n}|^2e^{-nQ}\, dA\le Ce^{-cn}\end{equation}
where $c$ is some positive constant.

Let $\pi_{j,n}$ be the orthogonal projection of the space $L^2(S)=L^2(S,dA)$ onto the subspace
$$\calW_{j,n}=\{pe^{-nQ^S/2};\, \degree (p)\le j\}\subset L^2(S),$$ consisting of weighted analytic polynomials. We will write
$w_{j,n}^\dagger=\pi_{j,n}[w_{j,n}^\sharp]$ where (as before) $w_{j,n}^\sharp=\chi_0\cdot F_{j,n}\cdot e^{-nQ^S/2}$. The estimate \eqref{ULK} implies that
\begin{equation}\label{show0}\begin{split}&\|w_{j,n}^\dagger-w_{j,n}^\sharp\|_2
\le Ce^{-c n}.\\
\end{split}\end{equation}


Combining \eqref{show0} with Corollary \ref{bom}, we see that
\begin{equation}\label{shave}\|w_{j,n}^\dagger\|_2=1+O(1/\sqrt{n}).\end{equation}

Also, by Lemma \ref{appop}, we have
for each $w\in\calW_{j-1,n}$ that
\begin{equation}\label{star}|\int_{S}w_{j,n}^\dagger\bar{w}\, |\le Cn^{-1/2}\cdot \|w\|_2.\end{equation}

Now let $w_{j,n}^\ddagger=w_{j,n}^\dagger-\pi_{j-1,n}[w_{j,n}^\dagger]$.

In view of \eqref{star} we have $\|\pi_{j-1,n}[w_{j,n}^\dagger]\|_{2}\le Cn^{-1/2}$, and hence
\begin{equation}\label{shave2}\|w_{j,n}^\ddagger-w_{j,n}^\dagger\|_2\le Cn^{-1/2},\qquad \|w_{j,n}^\ddagger\|_2=1+O(n^{-1/2}).\end{equation}

 Now write $w_{j,n}^\ddagger=c_{j,n}w_{j,n}$ where $w_{j,n}=p_{j,n}e^{-nQ^S/2}$ is the weighted orthonormal polynomial of degree $j$, and where
$c_{j,n}$ is a constant. We can assume that $c_{j,n}>0$.

The estimate \eqref{shave2} shows that $c_{j,n}=\|w_{j,n}^\ddagger\|_{2}=1+O(n^{-1/2})$, so
\begin{equation}\label{close}\|w_{j,n}^\ddagger-w_{j,n}\|_{2}=|c_{j,n}-1|\le C(n^{-1/2}).\end{equation}

All in all, the estimates \eqref{show0}--\eqref{close} imply
\begin{equation}\label{fest}\begin{split}\|w_{j,n}-w_{j,n}^\sharp\|_2&\le \|w_{j,n}-w_{j,n}^\ddagger\|_2+\|w_{j,n}^\ddagger-w_{j,n}^\dagger\|_2+
\|w_{j,n}^\dagger-w_{j,n}^\sharp\|_2\le
C(n^{-1/2}).
\end{split}\end{equation}

It is easy to see that the proof of Lemma \ref{l6} goes through also for our present weighted quasipolynomials, provided (say) that we restrict to the complement $\C\setminus X$.
We thus conclude from \eqref{fest} that
\begin{equation}\label{fest2}\|w_{j,n}-w_{j,n}^\sharp\|_{L^\infty(S\setminus X)}\le C.\end{equation}

Rescaling about the boundary point $0$ via \eqref{rescale}, we now conclude that, if $z$ stays in a fixed compact subset of $\C$ while $\zeta$
is given by \eqref{rescale},
\begin{align*}R_n(z)&\sim \frac 1 {n\Lap Q(0)}\sum_{j=n-\sqrt{n}\log n}^{n-1}|w_{j,n}(\zeta)|^2\\
&=\frac 1 {n\Lap Q(0)}\sum_{j=n-\sqrt{n}\log n}^{n-1}\left[|w_{j,n}^\sharp(\zeta)|^2+O(n^{1/4})\right]\\
&= H(z+\bar{z})\cdot\1_\L(z)+O(n^{-1/4}\log n).
\end{align*}
Here we used \eqref{sles}, \eqref{fest2}, and the bound $|w_{j,n}^\sharp|\le Cn^{1/4}$, which is obvious from the form of $w_{j,n}^\sharp$.

Thus
$R_n(z)\to H(z+\bar{z})\cdot \1_\L(z)$ where the convergence
is locally uniform on $\C\setminus (i\R)$ and locally bounded on $\C$ (Corollary \ref{goodb}), and we may
conclude our proof of Theorem \ref{mthm2}. $\qed$

\end{document}